\documentclass[a4paper,11pt,reqno]{amsart}

\usepackage{amssymb}
\usepackage{amsfonts}
\usepackage{amsmath}
\usepackage{amscd}
\usepackage{amsthm}
\usepackage{color}
\usepackage{mathtools}
\usepackage{tikz-cd}
\usepackage[pagebackref, hypertexnames=false, colorlinks, citecolor=red, linkcolor=blue, urlcolor=red]{hyperref}

\theoremstyle{plain}

\newtheorem{thm}{Theorem}[section]
\newtheorem*{thm*}{Theorem}
\newtheorem{cor}[thm]{Corollary}
\newtheorem{lem}[thm]{Lemma}
\newtheorem{prop}[thm]{Proposition}

\newtheoremstyle{named}{}{}{\itshape}{}{\bfseries}{.}{.5em}{#1 \thmnote{#3}}
\theoremstyle{named}

\numberwithin{equation}{section}

\DeclareSymbolFont{rsfs}{U}{rsfs}{m}{n}
\DeclareSymbolFontAlphabet{\mathscr}{rsfs}

\def\C{{\mathbb C}}

\def\R{\mathbb{R}}

\newcommand\eq[1] {(\ref{#1})}
\newcommand{\bfm}[1]{\mbox{\boldmath ${#1}$}}
\def\beq{\begin{eqnarray}}
\def\eeq{\end{eqnarray}}
\def\beqa{\begin{eqnarray*}}
\def\eeqa{\end{eqnarray*}}

\newcommand{\Real}{\mathop{\rm Re}\nolimits}
    \newcommand{\Imag}{\mathop{\rm Im}\nolimits}

\newcommand{\nonum}{\nonumber \\}
\newcommand{\bbeqa}{\begin{eqnarray}}
\newcommand{\eeeqa}[1]{\label{#1}\end{eqnarray}}
\newcommand{\bbeq}{\begin{equation}}
\newcommand{\eeeq}[1]{\label{#1}\end{equation}}

\newcommand{\Ga}{\alpha}

\newcommand{\Gs}{\sigma}

\newcommand{\BGs}{\bfm\sigma}

\newcommand{\BGS}{\bfm\Sigma}

\newcommand{\bpm}{\begin{pmatrix}}
\newcommand{\epm}{\end{pmatrix}}
\newcommand\fig[1] {{\rm Figure}~\ref{fig:#1}}

\newcommand\labfig[1] {\label{fig:#1}}

\usepackage{url}
\usepackage{hyperref}

\def\Im{{\it Im}}

\def\Be{{\bf e}}
\def\Bf{{\bf f}}
\def\Bg{{\bf g}}
\def\Bh{{\bf h}}

\def\BA{{\bf A}}
\def\BB{{\bf B}}
\def\BC{{\bf C}}
\def\BD{{\bf D}}

\def\BF{{\bf F}}

\def\BI{{\bf I}}

\def\BK{{\bf K}}

\def\BM{{\bf M}}

\def\BR{{\bf R}}
\def\BS{{\bf S}}

\begin{document}

\title[] {A theory of composites perspective on matrix valued Stieltjes functions}

\author{Graeme Milton}
\address[G.~Milton]{Department of Mathematics, The University of Utah, Salt Lake City, UT, USA}
\email{\tt milton@math.utah.edu}

\author{Mihai Putinar}
\address[M.~Putinar]{University of California at Santa Barbara, CA,
USA and Newcastle University, Newcastle upon Tyne, UK} 
\email{\tt mputinar@math.ucsb.edu, mihai.putinar@ncl.ac.uk}

\thanks{}
\dedicatory{}

\subjclass[2010]{30B70, 30E05, 42A38, 15A09, 74A40} 

\keywords{Composite material, effective conductivity, volume fraction, Stieltjes transform, Nevanlinna function, continued fraction}
\date  {}

\begin{abstract} A series of physically motivated operations appearing in the study of composite materials are interpreted in terms of elementary continued fraction transforms of matrix valued, rational Stieltjes functions.

\end{abstract}

\maketitle

\section{Introduction}
The aim of the present note is to sketch a link between selected works in the theory of composite materials and results in the class of rational functions of a complex variable
known as Stieltjes functions. To be more specific, results in the theory of composites showing that certain classes of microstructure, such as multicoated coated disk assemblages
and sequentially laminated materials, can achieve all possible conductivity functions of two-dimensional composites of two isotropic phases resonate to the finest detail to
traditional continued fraction expansions of (possibly matrix valued) Stieltjes functions. Bringing the relevant literature
on Stieltjes and Jacobi continued fraction expansions to the attention of the composite material community and vice-versa, drawing attention of the approximation theory experts to relevant developments in the theory of composites can be beneficial, both ways.

In a two phase isotropic composite of two isotropic conducting phases, the effective conductivity $\Gs^*$ is a function, $\Gs^*(\Gs_1,\Gs_2)$,
of the conductivities $\Gs_1$ and $\Gs_2$ of the two phases that is analytic except possibly when $\Gs_1/\Gs_2$ is zero, negative, or
infinite \cite{Bergman:1978:DCC, Milton:1981:BCP, Golden:1983:BEP}. (Complex conductivities govern the response of the current to
an electric field oscillating in time at fixed frequency.) Moreover,
the effective conductivity function satisfies: {\it homogeneity}:
\bbeq \Gs^*(\Ga\Gs_1,\Ga\Gs_2)=\Ga\Gs^*(\Gs_1,\Gs_2), \forall \Ga\in\C;
\eeeq{0.a}
the {\it Herglotz property}:
\bbeq \Real[\Gs^*(\Gs_1,\Gs_2)]>0,\text{  when  }\Real(\Gs_1)>0\text{ and }\Real(\Gs_2)>0,
\eeeq{0.b}
reflecting the fact that the composite absorbs electrical energy when both phases do; plus the
{\it normalization}:
\bbeq \Gs^*(1,1)=1, \eeeq{0.c}
reflecting the fact that the effective conductivity coincides with
that of the phases when they have equal conductivities.

  A corollary of the Herglotz and homogeneity properties (with complex $\Ga$) is that $\Gs^*(\Gs_1,\Gs_2)$ is real and positive when
$\Gs_1$ and $\Gs_2$ are both real and positive and
\bbeq  \Imag[\Gs^*(\Gs_1,\Gs_2)]>0,\text{  when  }\Imag(\Gs_1)>0\text{ and }\Imag(\Gs_2)>0.
\eeeq{0.d}

Owing to the homogeneity, we can without loss of generality set $\Gs_2=1$, $\Gs_1=z$, and $f(z)=\Gs^*(z,1)$. The function $f(z)$ satisfies
\bbeqa &~&\Imag[f(z)]\geq 0 \text{ when }\Imag(z)>0,\quad f(1)=1, \nonum
&~& \Imag[f(z)]=0,\text{ and } \Real[f(z)]>0 \text{ when }
\Imag(z)=0,\text{ and } \Real(z)>0. \nonum &~&
\eeeqa{0.e}
It turns out that $f(z)/z$ is a Stieltjes function and many results pertaining to Stieltjes functions provide valuable constraints on effective conductivities.
In particular, bounds for effective conductivities as surveyed in the book \cite{Milton:2002:TOC} follow from known bounds on Stieltjes functions.
Our note adds insight into this dictionary, by extracting from the applied context a general framework for analyzing matrix valued rational functions $\Bf(z)$ defined by the non-negativity of the kernels
$$ \frac{\Bf(z) - \Bf(z)^\ast}{z-\overline{z}} \geq 0,  \ \  \frac{z \Bf(z)^\ast - \overline{z} \Bf(z)}{z-\overline{z}} \geq 0, \ \ z \notin \R.$$
More specifically, Theorem \ref{reduction} below provides a succession of simple operations (translation, division and inversion)
which reduce a partial McMillan degree of $\Bf(z)$. This Stieltjes type continued fraction algorithm finishes in finitely many steps, proving that the composite material approach reaches out with the same efficiency to all such matrix valued rational functions $\Bf(z)$.

\section{A sketch of relevant results in the theory of composites}

The effective moduli of composites have been of interest for
centuries. They govern the macroscopic response, and homogenization
theory (see \cite{Bensoussan:1978:AAP} and the many other relevant
references in Chapter 1 of \cite{Milton:2002:TOC})
provides a rigorous mathematical foundation for their usage. As surveyed by Landauer \cite{Landauer:1978:ECI}, contributions to our initial
understanding of the conducting properties of isotropic composites are associated with works by Mossotti, Lorenz, Clausius, Lorentz, and Maxwell.
They provided justifications for the formula
\bbeq  \Gs^*=\Gs_2+\frac{c_1 \Gs_2}{(1-c_1)/3 -\Gs_2/(\Gs_2-\Gs_1)},
\eeeq{0.1}
of the effective conductivity $\Gs^*$ (or, equivalently, the effective dielectric constant, effective magnetic permeability,
effective thermal conductivity, effective fluid permeability, effective diffusion constant, effective shear modulus in antiplane elasticity)
of an isotropic  suspension of spheres of conductivity $\Gs_1$, occupying a volume fraction $c_1$ in a surrounding matrix of conductivity $\Gs_2$, occupying
a volume fraction $1-c_1$. Among other things the formula explains (in
the dielectric setting) the colors of glass containing metal particles
\cite{MaxwellGarnett:1904:CMG} as in old stained glass windows and as in the Roman Lycurgus Cup, now at the British museum.
Their analyses were only applicable when $c_1$ is small, as Rayleigh
demonstrated by considering periodic arrays of spheres \cite{Rayleigh:1892:IOA}.
Later Hashin and Shtrikman \cite{Hashin:1962:VAT}
 (see also \cite{Hashin:1962:EMH}), through an ingenious argument, showed that the formula is exact for all $c_1$ if the composite is an assemblage of coated spheres
 scaled to various sizes and packed to fill all space. Their observation, was that one could insert a single coated sphere
in a homogeneous material with conductivity
$\Gs^*$ without disturbing a surrounding uniform (constant) applied field. This can be continued until the coated spheres fill all space, and as the effective
conductivity is preserved at each insertion, $\Gs^*$ must be the
resultant effective conductivity of the two-phase mixture. We are free
to change the core material of each coated sphere to a microstructured material with
effective conductivity  $\Gs^*_1(\Gs_1,\Gs_2)$, and accordingly the
effective conductivity of the sphere assemblage becomes
\bbeq \Gs^*(\Gs_1,\Gs_2) =\Gs_2+\frac{c_1 \Gs_2}{(1-c_1)/3
  -\Gs_2/(\Gs_2-\Gs^*_1(\Gs_1,\Gs_2))}.
\eeeq{0.1a}
In particular, following ideas of Schulgasser \cite{Schulgasser:1977:CET},  $\Gs^*_1(\Gs_1,\Gs_2)$ could be the effective
conductivity of a second sphere assemblage, with cores of conductivity
$\Gs_2$ and coatings of conductivity $\Gs_1$, in which the core
material occupies a volume fraction $c_2$ in the second sphere
assemblage, so that
\bbeq \Gs^*_1(\Gs_1,\Gs_2)=\Gs_1+\frac{c_2 \Gs_1}{(1-c_2)/3 -\Gs_1/(\Gs_1-\Gs_2)},
\eeeq{0.2}
Furthermore, as the field was uniform
in each core of the first sphere assemble, one can replace each core by a
single coated sphere with a coating of conductivity $\Gs_1$ and a core of conductivity $\Gs_2$
occupying a volume fraction $c_2$ 
in that coated sphere. In this way we obtain an assemblage of doubly
coated spheres with effective conductivity
\bbeq  \Gs^*=\Gs_2+{c_1\Gs_2 \over\displaystyle (1-c_1)/3 -{\strut \Gs_2
\over\displaystyle\Gs_2-\Gs_1-{\strut c_2\Gs_1\over\displaystyle
(1-c_2)/3
-{\strut \Gs_1\over\displaystyle \Gs_1-\Gs_2 }}}}, 
\eeeq{0.3}
obtained by substituting \eq{0.2} in \eq{0.1a}, and implied by (7.16)
of \cite{Milton:2002:TOC}.

More generally for multicoated sphere assemblages one obtains a truncated
continued fraction for $\Gs^*$, and  $f(z)$ takes the form 
\bbeq f(z)=1+{c_1\over\displaystyle d_1 -{\strut 1
\over\displaystyle 1-z-{\strut c_2 z\over\displaystyle
d_2
-{\strut z\over\displaystyle z -1 -{\strut c_3 \over\displaystyle d_3
    - {\strut 1
\over\displaystyle 1-z-{\strut c_4 z\over\displaystyle
d_4 -\ldots}}}}}}}, 
\eeeq{0.3aa}
with $d_j=(1-c_j)/3$.

In the case of 2-dimensional composites and coated disk geometries (or equivalently
for the transverse effective conductivity of arrays of aligned multicoated cylinders) similar formulae hold, but with the
factors of $3$ in the preceding equations replaced by $2$. An example of a doubly coated disk geometry is illustrated in
\fig{1}. No matter what the geometry of an isotropic  2-dimensional composite, the effective conductivity
must satisfy Keller's phase interchange relation \cite{Keller:1964:TCC}
\bbeq \Gs^*(\Gs_1,\Gs_2)\Gs^*(\Gs_2,\Gs_1)=\Gs_1\Gs_2.
\eeeq{0.4}
It was stated, without proof, at the end of Section V in  \cite{Milton:1981:BTO} that any rational function
$\Gs^*(\Gs_1,\Gs_2)$ satisfying the required properties \eq{0.a},
\eq{0.b}, \eq{0.c} and \eq{0.4} is realizable by an appropriate assemblage
of multicoated disks and thus has a continued fraction expansion that
is an extension of the form \eq{0.3}. Thus multicoated disk geometries form a representative class for
all isotropic mixtures of two isotropic conducting phases having
rational functions  $\Gs^*(\Gs_1,\Gs_2)$. These can approximate, in a
suitable sense,  any function satisfying the required properties that is not rational.

\begin{figure}
	\centering
	\includegraphics[width=0.9\textwidth]{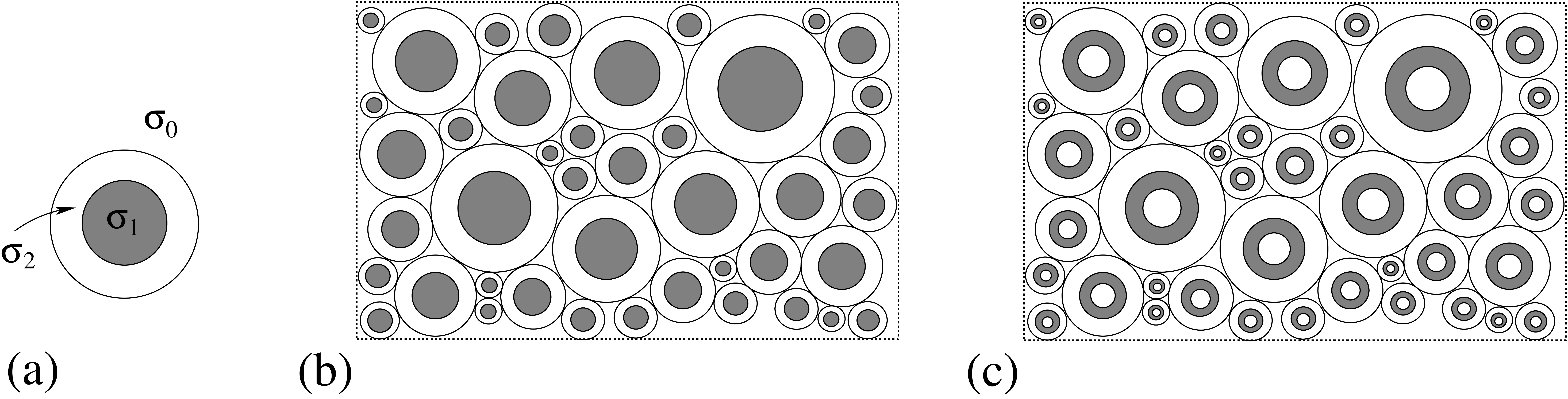}
	\caption{{\bf Steps in the construction of multicoated disk assemblages.} 
	{\bf (a)} In a homogeneous medium with conductivity $\Gs_0$
          given by the right hand side of \eq{0.1}, but with 3 replaced by 2, and in which there is a constant electric field, one may insert a coated disk, with a core of conductivity
          $\Gs_1$ occupying an area $c_1$ of the disk
          and having a shell of conductivity $\Gs_2$, without disturbing the surrounding constant field. 
           {\bf (b)} One can similarly insert into the medium with conductivity $\Gs_0$
          arbitrarily many rescaled copies of the coated disks in a periodic arrangement with the unit cell of periodicity marked by the dashed line. At each stage the average
          electric field and average current field remain unchanged, so the  effective conductivity $\Gs_*$ at any stage can be identified with the initial conductivity $\Gs_0$.
          The process can be continued until we have a coated disk assemblage with conductivity $\Gs_*=\Gs_0$ where the unit cell is completely filled by coated disks ranging
          to the infinitesimally small in size, and the medium with conductivity $\Gs_0$ is gone.  {\bf (c)}
          One may change the conductivity $\Gs_1$ inside the disk cores to another of conductivity $\Gs_0'$, with $\Gs_0'$ given
          by the right hand side of \eq{0.2}, with 3 replaced by 2, and accordingly adjust $\Gs_*$ to the value given by \eq{0.1a}, but with 3 replaced by 2 and with $\Gs_0'$ replacing $\Gs^*_1(\Gs_1,\Gs_2)$.
          As the field is also constant inside the cores of the disks, one may then replace these cores by
          coated disks having a core of conductivity $\Gs_2$ occupying a volume fraction $c_2$ of these disks and having
          a shell of conductivity $\Gs_1$. Ultimately, as the doubly coated disks fill all the unit cell, one obtains
          a doubly coated disk assemblage having conductivity given by \eq{0.3}, but with 3 replaced by 2. By repeating this construction one obtains multicoated
          disk assemblages with the disks having any
          desired number of shells.}
                 	\labfig{1}
                      \end{figure}

A proof in the more generalized setting of anisotropic composites,
where the scalar valued effective conductivity $\Gs^*(\Gs_1,\Gs_2)$ is replaced by the $2\times 2$ symmetric matrix valued effective conductivity
$\BGs^*(\Gs_1,\Gs_2)$ was provided later in \cite{Milton:1986:APLG} (see also Section
18.5 in \cite{Milton:2002:TOC}). In that setting \eq{0.a}, \eq{0.b}, and \eq{0.c} still hold with $\Gs^*(\Gs_1,\Gs_2)$ replaced by
$\BGs^*(\Gs_1,\Gs_2)$ and the inequalities holding in the sense of
quadratic forms. The phase interchange relation \eq{0.4} generalizes
\cite{Mendelson:1975:TEC}
to
\bbeq \BGs^*(\Gs_1,\Gs_2)\BR_\perp\BGs^*(\Gs_2,\Gs_1)\BR_\perp^T=\Gs_1\Gs_2\BI,
\eeeq{0.5}
in which
\bbeq \BR_\perp=\bpm 0 & 1 \\ -1 & 0 \epm,
\eeeq{0.6}
with transpose $\BR_\perp^T=-\BR_\perp$, is the matrix for a
$90^\circ$ rotation. With $\BI$ denoting the identity tensor and  a possibly anisotropic effective
tensor $\BGs^*_1$ taking the place of $\Gs^*_1$, \eq{0.1a} is replaced
by Tartar's formula \cite{Tartar:1985:EFC} (that generalized a result
of Maxwell \cite{Maxwell:1954:TEM}),
\bbeq \BGs^*=\Gs_2\BI+c_1 \Gs_2\left[(1-c_1)\BM_1 -\Gs_2\BI(\Gs_2\BI-{\BGs}^*_1)^{-1}\right]^{-1},
\eeeq{0.7}
for the effective conductivity tensor of a hierarchical laminate, where layers of phase 2 have been inserted into the material with conductivity  $\BGs^*_1$,
at various angles and in various proportions at widely separated length scales. Here phase 2 occupies a volume fraction $1-c_1$ in the hierarchical laminate
and $\BM_1$ depends on the angles and proportions of laminations and
can be any positive semidefinite symmetric real matrix with trace 1. 

The key idea in making the correspondence in the isotropic case is to observe that in the multicoated coated disk geometry, with effective conductivity such as that
given by \eq{0.3} one can extract the value of $c_1$ by setting
$\Gs_1=0$. This makes physical sense as a non-conducting material
shields any structure inside it.  Knowing this one can recover
$\Gs^*_1(\Gs_1,\Gs_2)$ from \eq{0.1a} or more generally the effective
conductivity of the multicoated coated disk geometry
with the outermost shell removed from each disk. A similar operation can be done for
any rational function $\Gs^*(\Gs_1,\Gs_2)$ satisfying the
required properties to obtain a function $\Gs^*_1(\Gs_1,\Gs_2)$
with $\Gs^*_1(0,\Gs_2)=0$.

At the next step one sets $\Gs_2=0$ to recover
$c_2$ from  ${\Gs}_1^*(\Gs_1,\Gs_2)$. Iterating the procedure reduces
the degree of the rational function until after $N$ iterates one ends
up with the trivial functions  $\Gs^*_N(\Gs_1,\Gs_2)=\Gs_1$ or
$\Gs^*_N(\Gs_1,\Gs_2)=\Gs_2$.
The same sort of analysis applies in the anisotropic case.

Our objective is to carry through the same sort of analysis, condensed
to a concise form, to Stieltjes and related functions. We will generalize it to avoid using
relations like \eq{0.4} and \eq{0.5} and identities satisfied by the
constants (or constant matrices) entering the continued fraction expansion.

We note in passing that there is a completely different class of two-dimensional microgeometries, as depicted in \fig{2} that realizes the isotropic conductivity function $\Gs^*(\Gs_1,\Gs_2)$ satisfying \eq{0.4}. Following appendix B in \cite{Milton:1981:BTO}, or Section 18.6 in \cite{Milton:2002:TOC}, the first
step is to make a total of $m+2$ simple laminates: laminate $\Ga$ for $\Ga=0,1,2\ldots,m+1$ is comprised of phase 1 laminated with
phase 2, in proportions $q_\Ga$ and $1-q_\Ga$ respectively. The laminates are oriented so the applied field $\Be_0$ is perpendicular to the interfaces.
Hence their conductivity in this direction is the harmonic average $\Gs^{*L}_\Ga=[q_\Ga/\Gs_1+(1-q_\Ga)/\Gs_2]^{-1}$. The second step is to layer together these laminates, on a much larger length scale, in proportions $a_0,a_1,\ldots,a_m$ satisfying
\bbeq \sum_{\Ga=0}^{m+1}a_\Ga=1, \eeeq{0.9a}
with the applied field  $\Be_0$ parallel to the new interfaces. The effective conductivity
in the $\Be_0$ direction is then the arithmetic average of the  $\Gs^{*L}_\Ga$:
\bbeq \Gs^*_{\parallel}(\Gs_1,\Gs_2)=\sum_{\Ga=0}^{m+1}a_\Ga\Gs^{*L}_\Ga=\sum_{\Ga=0}^{m+1}a_\Ga[q_\Ga/\Gs_1+(1-q_\Ga)/\Gs_2]^{-1}.
\eeeq{0.9}
According to \eq{0.5} the effective conductivity in the perpendicular direction is
\bbeq \Gs^*_{\perp}(\Gs_1,\Gs_2)=\Gs_1\Gs_2/\Gs^*_{\parallel}(\Gs_2,\Gs_1). \eeeq{0.10}
In particular, given any rational function $\Gs^*(\Gs_1,\Gs_2)$ satisfying \eq{0.a}, \eq{0.b}, and \eq{0.c} then one can find constants
$a_\Ga\geq 0$ satisfying \eq{0.9a}  and $q_\Ga\in [0,1]$ such that $\Gs^*(\Gs_1,\Gs_2)$ can be expressed in the form \eq{0.9}, i.e.
$\Gs^*(\Gs_1,\Gs_2)$ can be identified with the effective conductivity $\Gs^*_{\parallel}(\Gs_1,\Gs_2)$ in the  $\Be_0$ direction of the
laminate geometry of \fig{2}. If, in addition, $\Gs^*(\Gs_1,\Gs_2)$ satisfies \eq{0.4} then \eq{0.9} implies
$\Gs^*_{\parallel}(\Gs_1,\Gs_2)=\Gs^*_{\perp}(\Gs_1,\Gs_2)=\Gs^*(\Gs_1,\Gs_2)$. In summary, microgeometries of the type illustrated in \fig{2}
suffice to generate any desired rational function $\Gs^*(\Gs_1,\Gs_2)\BI$ as their effective conductivity tensor, provided only that
$\Gs^*(\Gs_1,\Gs_2)$ satisfies \eq{0.a}, \eq{0.b}, \eq{0.c}, and \eq{0.9}. However, these geometries do not suffice to generate any
matrix valued function satisfying \eq{0.5} as for these microstructures the real and imaginary parts of $\BGs^*(\Gs_1,\Gs_2)$ always commute,
which is not necessarily the case.

We remark that there is a separate route to establishing the correspondence with these hierarchical laminates. Instead of working with the analytic functions one
can extend the approach in \cite{Clark:1994:MEC} where a finite dimensional approximation to the underlying Hilbert space and relevant subspaces is found to be isomorphic to the
Hilbert space and relevant subspaces of these hierarchical layered geometries. This will be explained in more detail in a subsequent paper.

\begin{figure}
	\centering
	\includegraphics[width=0.9\textwidth]{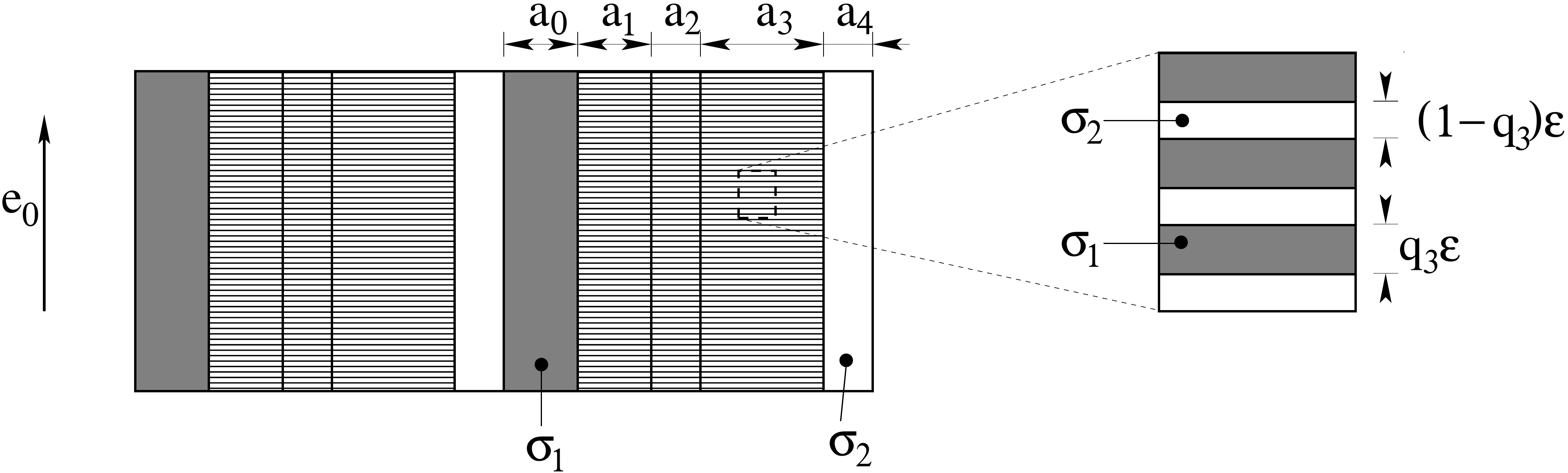}
	\caption{The unit cell of periodicity of a two-dimensional geometry, that can achieve any function $\BGs^*(\Gs_1,\Gs_2)$
        remaining diagonal for all $\Gs_1$ and $\Gs_2$.}
                 	\labfig{2}
                      \end{figure}

\section{The reduction algorithm}

Let $n \geq 1$ be an integer and denote by $M_n(\C)$ the algebra of $n \times n$ matrices over the complex field.
An element $\BA \in M_n(\C)$ is called non-negative, denoted $\BA \geq 0$, if $\BA$ is self-adjoint with non-negative eigenvalues.
The matrix $\BA$ is called positive, denoted $\BA>0$, if it is self-adjoint and all eigenvalues are positive.

Let ${\mathcal G}$ denote the class of $M_n(\C)$-valued rational functions of the form:
\begin{equation}\label{def}
 \Bf(z) = \BA z +\BB - \sum_{j=1}^d \frac{\BC_j}{z+ \lambda_j},
 \end{equation}
where $\lambda_j > 0, \ 1 \leq j \leq d,$ are positive real numbers, the matrices $\BA, \BB, \BC_j \in M_d(\C)$ are non-negative
and
$$ \Bf(0) = \BB - \sum_{j=1}^d  \frac{\BC_j}{\lambda_j} \geq 0.$$

Such a function belongs to Nevanlinna class, that is 
$$ \Im \Bf(z) \Im (z) \geq 0, \ \ z \in {\mathbf C} \setminus \{ -\lambda_1, -\lambda_2, \ldots, -\lambda_d \}.$$
Since the poles are situated on the negative semi-axis, $\Bf(z)$ may also be called a Stieltjes function. The assumption 
 $\Bf(x) \geq 0$ (with respect to operator ordering) on the positive semi-axis $x \geq 0$, is however non redundant.

Note that $\mathcal G$ is a convex cone. The class $\mathcal G$ can be defined with values linear operators in a complex Hilbert space
$H$; in that case we adopt the notation ${\mathcal G}(H)$.
We refer to Appendix A in \cite{BHdS} for a detailed discussion of integral representations of 
operator valued Nevanlinna functions.

A number of elementary transforms leave the class $\mathcal G$ invariant.

\begin{lem}The reflection $\Bf \mapsto (R\Bf)(z) = z \Bf(\frac{1}{z})$ maps $\mathcal G$ to itself.
\end{lem}

\begin{proof}
Indeed, denote $\Bg(z) = z \Bf(\frac{1}{z}),$ where $\Bf$ has the representation (\ref{def}). Then
\begin{equation}\label{flip}
 \Bg(z) = [ \BB -  \sum_{j=1}^d  \frac{\BC_j}{\lambda_j}] z + [\BA +  \sum_{j=1}^d  \frac{\BC_j}{\lambda_j^2}] -  \sum_{j=1}^d  \frac{\frac{\BC_j}{\lambda_j^3}}{z+ \frac{1}{\lambda_j}}.
 \end{equation}
Note that $$ \lim_{z \rightarrow \infty} \frac{\Bg(z)}{z} = \Bf(0).$$ On the other hand, if $\Bf(\infty)$ is finite (that is $\BA=0$), then $\Bg(0) = 0$.
\end{proof}

\begin{lem} A translation $\Bf \mapsto (T_S \Bf)(z) = \Bf(z) + \BA_1 z + \BB_1$, where  $\BA_1 \geq 0$ and $\BB_1= \BB_1^\ast \geq - \Bf(0),$
leaves the class $\mathcal G$ invariant.
\end{lem}

More intricate is the inversion, this time restricted to elements outside the class $\mathcal G$ and touching a Moore-Penrose inverse operation.

\begin{lem} \label{inversion} Let
$$ \BR(z) = -\sum_{j=1}^d \frac{\BC_j}{z+ \lambda_j},$$
where $0< \lambda_1 < \ldots \lambda_d$ and $\BC_j \geq 0, \ 1 \leq j \leq d.$
Denote by $\Pi$ the orthogonal projection of the ambient space $\C^n$ (where the linear transforms $\BC_j$ act) onto the range
of the non-negative operator $\BC_1 + \BC_2 + \ldots + \BC_d$.
Then the matrix valued rational function
$$ -[\BR]^{-1}(z)x  = \left\{
          \begin{array}{lr}
          -(\Pi \BR(z) \Pi)^{-1}x, &  \Pi x = x,\\
          0, & \Pi x = 0,
          \end{array} \right.$$
          belongs to $\mathcal G$.
          \end{lem}

\begin{proof}
Let $\BC = \BC_1 + \BC_2 + \ldots + \BC_d$ and define $V = \Pi (\C^n)$.
The rational function $\BR$ maps $V$ to itself, and moreover, its compression
$$ \BS(z) = \Pi \BR(z) \Pi : V \longrightarrow V,$$
 is invertible for large values of $z$:
 $$ \BS(z) = -\frac{\Pi \BC \Pi}{z} + \sum_{j=1}^d \frac{\lambda_j \Pi \BC_j \Pi}{z^2} + \ldots, \ \ |z| >> 1.$$
We claim that the function $ \BS(z)^{-1} $ belongs to the class $\mathcal G(V)$, referring to rational functions with values linear operators acting on the space $V$.
In the end we will define 
$$ [\BR]^{-1} = \BS(z)^{-1} \oplus 0,$$
with respect to the decomposition $\C^n = V \oplus V^\perp$.

Note that the rational function $\BS(z)$ has isolated zeros, since its determinant is non-vanishing at infinity. Hence $\BS(z)^{-1}$ is well defined on the space $V$, and it is an 
operator valued rational function.
In addition, the identity
$$ \frac{\BS(z) - \BS(z)^\ast}{z-\overline{z}}  = \sum_{j=1}^d \frac{\lambda_j \Pi \BC_j \Pi}{|z+ \lambda_j|^2}$$
shows that $\BS(z)$ has zeros only on the negative semi-axis and $\BS(z)$ is a Nevanlinna function.

Moreover, Loewner's transform $\BS \mapsto -\BS^{-1}$ yields
$$ -\BS(z)^{-1} + [\BS(z)^\ast]^{-1} = [\BS(z)^\ast]^{-1} [ \BS(z) - \BS(z)^\ast] \BS(z)^{-1}.$$
Consequently 
$$ \Im [  -\BS(z)^{-1}] \Im z > 0, \ \ z \notin \R,$$
that is $-\BS(z)^{-1}$ is also a Nevanlinna function. In addition, 
$$ \lim_{z \rightarrow \infty} z \BS(z) = -\BC $$
is an invertible operator on $V$, hence
$$ -\BS(z)^{-1} = \BC^{-1} z + \BS_1(z),$$
where $\BS_1(z)$ is a rational function belonging to Nevanlinna class, this time finite at infinity. Note also that
$$\BS(0) = - \sum_{j=1}^d \lambda_j^{-1} \BC_j < 0,$$
that is $-\BS(0)^{-1} >0$. Finally, for every $x \geq 0,$ the operator
$$ -\BS(x) = \sum_{j=1}^d \frac{\BC_j}{x+ \lambda_j} $$
is self-adjoint and strictly positive. Consequently $-\BS(x)^{-1}$ is invertible for all $x \geq 0$.
In conclusion the operator valued function $-\BS(z)^{-1}$ has a pole at infinity, plus possibly other poles, all located on the negative semi-axis.

\end{proof}

Next we identify a succession of such elementary operations which will reduce the number of finite poles of a rational function belonging to the class $\mathcal G$.

To this aim, a relevant quantitative indicator is the partial McMillan degree of an element $\Bf \in \mathcal G$:
$$ \delta(\Bf) = \sum_{j=1}^d {\rm rank}\  \BC_j.$$
The McMillan degree of $\Bf$ is ${\rm rank A} + \delta(\BF)$, see Corollary 2 in \cite{Duffin}.
In other terms, the indicator $\delta(\Bf)$ collects the contributions to McMillan's degree of all matricial poles of $\Bf$ except the point at infinity.

\begin{lem} \label{delta} Let $ \BR(z) = -\sum_{j=1}^d \frac{\BC_j}{z+ \lambda_j},$ be as in Lemma \ref{inversion}.
Then $\delta [\BR]^{-1} < \delta \BR.$
\end{lem}

\begin{proof}
We are under the assumption $\delta \Bf >0$, that is not all matrix weights $\BC_j$ are null. Let $\Pi$ denote the orthogonal projection appearing in the proof of the inversion lemma.
The rational function $\BS(z) = \Pi \BR(z) \Pi$ has a zero of order
$\sigma  = {\rm rank} (\BC_1+\ldots+\BC_d) > 0$ at infinity. Due to the invariance of the McMillan degree to inversion (Theorem 3 in \cite{Duffin}) we find
$$ \delta \BR = \kappa_1 + \ldots + \kappa_d = \sigma + \delta [ \BR]^{-1}.$$
\end{proof}

Originally, the algorithm derived on composite material grounds runs as follows.
Suppose there is a pole at infinity, that is $\BA \neq 0$. Then the function $\Bg_1(z) = z[ \Bf(\frac{1}{z})-\Bf(0)]$ has the same number of finite poles, but no pole at infinity. Identity (\ref{flip})
implies $\delta \Bf = \delta \Bh.$ 

Next, assuming $\BA=0$, that is $\Bf(\infty) = \BB \geq \Bf(0)$, one performs an inversion, possibly well defined only on a closed subspace of $\C^n$.
Specifically, invert (partially) the rational function of the form
$$ \Bg_1(z) - \Bg_1(\infty) = -\sum_{j=1}^d \frac{\BD_j}{z+ \mu_j}, \ \ \mu_j >0, \ \BD_j \geq 0.$$
Putting together, in the spirit of continued fractions the above two steps, one finds
$$ \Bf(z) = \Bf(0) + z \Bg_1(\frac{1}{z}),$$
and
$$ \Bg_1(z) - \Bg_1(\infty) = -[\Bg_2(z)]^{-1}.$$
Starting with $\Bf$ of the form (\ref{def}) we infer:
$$ \Bf(z) = \BB- \sum_{j=1}^d \frac{\BC_j}{\lambda_j} + z  \Bg_1(\frac{1}{z}),$$
that is
$$ \Bg_1(\frac{1}{z}) = \BA + \sum_{j=1}^d \frac{\BC_j}{\lambda_j (z + \lambda_j)},$$
implying
$$\Bg_1(\infty) = \BA + \sum_{j=1}^d \frac{\BC_j}{\lambda_j^2}.$$
Summing up the two steps:
$$ \Bf(z) = \Bf(0) + \Bf'(0) z - z [\Bg_2(\frac{1}{z})]^{-1},$$
where we have identified $\Bf(0) = \BB- \sum_{j=1}^d \frac{\BC_j}{\lambda_j} $ and $\Bf'(0) = \BA + \sum_{j=1}^d \frac{\BC_j}{\lambda_j^2}.$
More exactly,
$$ z [\Bg_2(\frac{1}{z})]^{-1} = z^2 \sum_{j=1}^d \frac{\BC_j}{\lambda_j^2(z+\lambda_j)},$$
or equivalently
$$ [\Bg_2(z)]^{-1} = \sum_{j =1}^d  \frac{\BC_j}{\lambda_j^2(1+\lambda_j z)}.$$

Putting together the above constructive approach we state the main result translating the composite material inspired algorithm.

\begin{thm}\label{reduction} Let $\Bf(z)$ be a matrix valued rational function belonging to the class $\mathcal G$. Then
\begin{equation}\label{mainstep}
\Bf_1(z) = -[\frac{\Bf(z) - \Bf(0) - \Bf'(0)z}{z^2}]^{-1} 
\end{equation}
is a rational function in $\mathcal G$, and $\delta \Bf_1 < \delta \Bf$.
\end{thm}

\begin{proof}
Let $\Bf$ be represented as in (\ref{def}). Then
$$ - \frac{\Bf(z) - \Bf(0) - \Bf'(0)z}{z^2} = \sum_{j=1}^d \frac{\BC_j}{\lambda_j^2 (\lambda_j + z)},$$
and we apply Lemma \ref{delta}.
\end{proof}

By a slight abuse of notation, writing the Moore-Penrose inverses of matrices as fractions, one obtains
$$ \Bf(z) = \BA_0 + \BB_0 z - \frac{z^2}{ \BA_1 + \BB_1 z - \frac{z^2}{\BA_2 + \BB_2 z - \frac{z^2}{\ddots}}},$$
where all entries $\BA_j, \BB_j$ are non-negative matrices.
The reflection transform puts this fraction into the classical form
$$ z \Bf(\frac{1}{z}) = \BA_0 z + \BB_0 - \frac{1}{\BA_1 z + \BB_1 - \frac{1}{\BA_2 z + \BB_2  - \frac{1}{\ddots}}},$$
known, in the scalar case, as a J-fraction, see \cite{Wall} Theorem 23.3.

In this format, the reduction algorithm is even simpler: subtract the pole at infinity, then subtract the value at infinity, then invert.
While this is the classical succession of operations transforming a meromorphic Laurent series into a continued fraction
(see for instance Chapter 8 in \cite{Perron}, or Chapter 9 in \cite{Wall}), it is notable for the present note that the class $\mathcal G$
of matrix valued rational 
functions remains invariant under this procedure. But even this observation is not new, essentially going back to the foundational work of Stieltjes
\cite{Stieltjes}.

\section{Stieltjes function interpretation}

Let $\Bf \in {\mathcal G}$ be a matrix valued rational function as described in the previous section, specifically of the form (\ref{def}), with the specific constraints.
The rational function
$$ \Bh(z) = \frac{\Bf(z)}{z} = \BA + [\BB-\sum_{j=1}^d \frac{\BC_j}{\lambda_j}]\frac{1}{z}  + \sum_{j=1}^d \frac{\BC_j}{\lambda_j} \frac{1}{z+\lambda_j}$$
can be written as
\begin{equation}\label{Stieltjes}
 \Bh(z) = \BA + \sum_{j=0}^d \frac{\BA_j}{z+ \lambda_j},
 \end{equation}
where all matrices $\BA, \BA_0, \BA_1, \ldots, \BA_d$ are non-negative and $\lambda_0 =0$. Equivalently, there exists a positive, matrix valued measure $d\BGS$ supported on the finite set
$\{ 0, \lambda_1, \ldots, \lambda_d\}$, with the property
$$ \Bh(z) = \BA + \int_0^\infty \frac{ d \BGS(t)}{t+z}.$$
The imaginary part of this function is non-positive in the upper half-plane:
$$ \Bh(z) - \Bh(z)^\ast =  \int_0^\infty \frac{ (\overline{z}-z)d \BGS(t)}{|t+z|^2},$$
and moreover, $\Bh$ is analytic on the semiaxis $(0, \infty)$ and $\Bh(x) \geq 0$ whenever $x>0$. Such a matrix-valued, or operator valued function is called a {\it Stieltjes function},
see for instance Appendix A in \cite{BHdS}.

The above reasoning can be reversed, leading to the following simple observation.

\begin{prop} A matrix valued rational function $\Bf$ belongs to the class $\mathcal G$ if and only if $\Bh(z) = \frac{\Bf(z)}{z}$ is a Stieltjes function.
\end{prop}

Indeed, starting with the Stieltjes function (\ref{Stieltjes}) one finds
$$ z \Bh(z)  = \BA z + \sum_{j=0}^d \frac{\BA_j z}{z+ \lambda_j} = 
z\BA +  \sum_{j=0}^d  \BA_j -\sum_{j=1}^d  \frac{\BA_j \lambda_j}{z+ \lambda_j}.$$
Note that $\lambda_0 = 0$ and $(z\Bh(z))|_{z=0} = \BA_0 \geq 0.$

The known characterization of matrix valued Stieltjes functions (see for instance \cite {Bo}) provides the following point-wise positivity certificate of our class of functions.

\begin{cor} A rational function $\Bf(z)$ belongs to the class $\mathcal G$ if and only if
$$ \frac{\Bf(z) - \Bf(z)^\ast}{z-\overline{z}} \geq 0, \ \ {\it and} \ \  \frac{ \overline{z} \Bf(z) - z \Bf(z)^\ast}{z-\overline{z}} \leq 0$$
for all $z \neq \overline{z}$.
\end{cor}

\begin{cor} Let $\Bf \in {\mathcal G}$ be represented as 
$$ \Bf(z) = z [ \BA + \int_0^\infty \frac{ d\BGS(t)}{t+z}],$$
where $d\BGS$ is an positive matrix valued measure.
Then
\begin{equation}\label{pointmass}
\Bf(z) = \BA z + \BGS ( [0,\infty)) -  \int_0^\infty \frac{ d\BGS_1(t)}{t+z}
\end{equation}
with the positive measure $d\BGS_1(t) = t d\BGS(t)$ without point mass at $t=0$, and 
 $$ \Bf(0) = \BGS(\{ 0 \}) = \BGS( [0,\infty)) -\BGS_1 ( [0,\infty)).$$
 \end{cor}

In virtue of Naimark's Dilation Theorem, one can write the semi-spectral measure as a compression of a spectral measure, to the extent that a function $\Bf \in \mathcal G$ can be written
as
$$ \Bf(z) = z [\BA + \BK^\ast (\BS+z\BI)^{-1} \BK],$$
where $\BA \geq 0$, $\BS \geq 0$ is a non-negative transform acting in a possibly larger finite dimensional Hilbert space, and $\BK$ is a linear transform.

Specializing to the scalar case, we owe to Stieltjes the following characterization of the continued fraction expansion of functions belonging to the class $\mathcal G$. 

\begin{thm}\label{CD} Let $f \in \mathcal G$ be a scalar function. Then
$$f(z) = a_0  z+ b_0 - \frac{1}{ a_1 z+ b_1 - \frac{1}{a_2z + b_2  - \frac{1}{\ddots -\frac{1}{a_d z + b_d}}}},$$
where $a_0, b_0 \geq 0; \ a_j, b_j >0, \ j \geq 1.$
\end{thm}

\begin{proof} In view of decomposition (\ref{pointmass}) we can write
$$ f(z) = a_0 z + b_0 -  \int_0^\infty \frac{ d\sigma_1(t)}{t+z},$$
with a finite positive measure $\sigma_1$ supported by $(0,\infty)$.
If $\sigma_1 =0$, then $f(z)$ is an affine function and there is nothing else to check.

In case the measure $\sigma_1$ is non-zero, Stieltjes fundamental theorem \cite{Stieltjes} pg. J3, yields
$$  \int_0^\infty \frac{ d\sigma_1(t)}{t+z} = \frac{1}{ a_1 z+ b_1 - \frac{1}{a_2z + b_2  - \frac{1}{\ddots -\frac{1}{a_d z + b_d}}}},$$
where all entries $a_j, b_j$ are positive.
\end{proof}

Notice that the parameters $a_j, b_j$ entering into the above continued fraction are not independent. For instance the constraint $f(0) \geq 0$ translates into:
$$ b_0 \geq  \frac{1}{b_1 - \frac{1}{b_2  -\frac{1}{b_3 - \frac{1}{\ddots}}}},$$
 in case $d \geq 1$. 
 
It is worth returning to the original reasoning in \cite{Stieltjes} and clarify the necessary constraints imposed on the parameters $a_j, b_j$.
Specifically, Stieltjes starts with the transform of a positive measure on $[0,\infty)$ of finite mass:
\begin{equation}\label{free}
  F(z) = \int_0^\infty \frac{ d\sigma (t)}{t+z} = \frac{1}{c_1 z+ \frac{1}{c_2 + \frac{1}{c_3 z + \frac{1}{c_4 + \frac{1}{\ddots}}}}},
  \end{equation}
and proves that the positive entries $c_j >0$ are free, that is independent of each other. In other terms, any finite selection of positive $c_j's$ 
produces a rational functions of the form $F(z)$. The convergence analysis of the continued fraction corresponding to an infinite choice of parameters $c_j$ is one of the major achievements of
 Stieltjes' celebrated memoir. 

Furthermore, a natural operation on continued fractions, called contraction, leads to
the equivalent representation
$$ F(z) = \frac{d_0}{z + d_1 - \frac{d_1 d_2}{z + d_2 + d_3 - \frac{d_3 d_4}{z+ d_4 + d_5 - \ddots}}},$$
with independent positive parameters $d_j >0, \ \ j \geq 0.$ This obviously imposes some relations among the successive entries of the fraction, such as 
$ d_1 >  \frac{d_1 d_2}{d_2 + d_3},$ and so on, see \cite{Stieltjes}. These constraints were linked by Stieltjes to a sequence of positivity certificates involving Hankel determinants.
At this point we touch the well known solvability conditions of a Stieltjes moment problem, cf. \cite{Perron,Wall}. 

In conclusion, if we aim at having a continued fraction 
decomposition with independent parameters of an element $f \in {\mathcal G}$, we better adopt the following scheme:
$$ f(z) = {\rm Res}_\infty (f) z + f_1(z),$$
$$ f_1(z) = f_1(\infty) - f_2(z),$$
$$ f_2(z) =  \frac{1}{c_1 z+ \frac{1}{c_2 + \frac{1}{c_3 z + \frac{1}{c_4 + \frac{1}{\ddots}}}}},$$
with unrestricted $c_j >0$, except $f(0) \geq 0$.

We do not expand here the many facets and applications of Stieltjes functions and their continued fraction expansions, see the informative surveys \cite{vA,VvA}. For the purpose of our note we mention that
matrix valued generalizations of Stieltjes continued fractions have been thoroughly studied from the point of view of operator valued, positive definite hermitian forms and their associated orthogonal polynomials \cite{AKvI,SvI,vI}. These cited articles offer a comprehensive approximation theory perspective on the entangled nature of matrix valued Stieltjes functions.

In a rather different direction,
operator valued Stieltjes functions were interpreted by M. Livsic since mid 1960-ies as impedance functions of certain open systems. The glorious past and present of this ever inspiring subject is well exposed in \cite{ABT}. Bounded analytic interpolation is one of the main beneficiaries of Livsic' framework \cite{Bo}. Third, but not last, regarding an operator valued Stieltjes function as a compressed resolvent of a non-negative operator is a central topics in control theory and spectral theory; a particular instance being the analysis of the Weyl-Titchmarch function of singular boundary value problems of Sturm-Liouville type, cf. \cite{BHdS}.

\end{document}